\documentclass{llncs}
\usepackage[english]{babel}
\usepackage{xcolor}
\usepackage{amsmath}
\usepackage{blindtext}
\usepackage{extarrows}
\usepackage{colonequals}
\usepackage{times}
\usepackage{hyperref}
\usepackage{subcaption}
\usepackage{tikz}
\usetikzlibrary{matrix,arrows.meta,shapes,snakes,automata,backgrounds,petri,positioning,er,arrows,decorations.text}
\input{tikzconfig}
\usepackage[normalem]{ulem}
\usepackage{cite}
\title{Relational to RDF Data Exchange in Presence of a Shape Expression Schema}

\author{Iovka Boneva, Jose Lozano, S\l{}awek Staworko}

\institute{Univ. Lille, CNRS, Centrale Lille, Inria, UMR 9189 - CRIStAL - Centre de Recherche en Informatique Signal et Automatique de Lille, F-59000 Lille, France}

\newcommand{\Rsig}{\mathcal{R}}

\newcommand{\Fsig}{\mathcal{F}}
\newcommand{\FOsig}{\mathcal{W}}
\newcommand{\Tsig}{\mathcal{T}}
\newcommand{\GTsig}{\mathcal{G}_\Tsig}
\newcommand{\Vsig}{\mathcal{V}}
\newcommand{\DES}{\mathcal{E}} 

\newcommand{\Model}{M}
\newcommand{\Inst}{I} 
\newcommand{\Ginst}{J} 

\newcommand{\Rsch}{\mathbf{R}} 
\newcommand{\ShexSch}{\mathbf{S}} 
\newcommand{\Vars}{V}
\newcommand{\Fint}{{F_{\mathrm{int}}}}
\newcommand{\fvars}{\operatorname{\mathit{fvars}}}

\newcommand{\vect}[1]{\mathbf{#1}}

\newcommand{\IriSet}{\mathbf{Iri}}

\newcommand{\NullSet}{\mathbf{Null}}
\newcommand{\NullLitSet}{\mathbf{NullLit}}

\newcommand{\BlankSet}{\mathbf{Blank}}
\newcommand{\LitSet}{\mathbf{Lit}}

\newcommand{\Dom}{\mathbf{Dom}}

\newcommand{\LitSym}{\mathit{Lit}}
\newcommand{\Triple}{\mathit{Triple}}

\newcommand{\typing}{\mathit{typing}}

\newcommand{\tconstr}{\mathit{tc}}
\newcommand{\multmostone}{\mathit{mult^{\mathord{\le}\!1}}}
\newcommand{\multleastone}{\mathit{mult^{\mathord{\ge}\!1}}}

\newcommand{\interval}[1]{{\ensuremath{\mathord{\text{\normalfont\fontfamily{lmtt}\selectfont{}#1}}}}}

\newcommand{\ONE}{\interval{1}\xspace}
\newcommand{\MAYBE}{\interval{?}\xspace}
\newcommand{\MANY}{\interval{*}\xspace}
\newcommand{\PLUS}{\interval{+}\xspace}

\newcommand{\fdep}{\Sigma_{\mathrm{fd}}}
\newcommand{\stdep}{\Sigma_{\mathrm{st}}}
\newcommand{\shexdep}{\Sigma_{\ShexSch}}

\newcommand{\chasestep}[4]{#1 \smash{\xrightarrow{#2,#3}} #4}

\newcommand{\dbl}{\mathbin{::}}

\newcommand{\dom}{\mathit{dom}}%
\newcommand{\ran}{\mathit{ran}}%

\newcommand{\fdprop}{\operatorname{\textsc{fdprop}}}

\newcommand{\fd}{\mathit{fd}}

\newcommand{\Bug}{\ensuremath{\mathsf{TBug}}\xspace}
\newcommand{\User}{\ensuremath{\mathsf{TUser}}\xspace}
\newcommand{\Emp}{\ensuremath{\mathsf{TEmp}}\xspace}
\newcommand{\Test}{\ensuremath{\mathsf{TTest}}\xspace}
\newcommand{\col}{\mathord{\kern-1pt\interval{:}\kern-1.25pt}}
\newcommand{\mel}{\ensuremath{\col\mathsf{email}}\xspace}
\newcommand{\name}{\ensuremath{\col\mathsf{name}}\xspace}
\newcommand{\descr}{\ensuremath{\col\mathsf{descr}}\xspace}
\newcommand{\rep}{\ensuremath{\col\mathsf{rep}}\xspace}
\newcommand{\repr}{\ensuremath{\col\mathsf{repro}}\xspace}
\newcommand{\rel}{\ensuremath{\col\mathsf{related}}\xspace}
\newcommand{\phone}{\ensuremath{\col\mathsf{phone}}\xspace}
\newcommand{\grp}{\ensuremath{\col\mathsf{covers}}\xspace}
\newcommand{\prep}{\ensuremath{\col\mathsf{prepare}}\xspace}
\newcommand{\RBug}{\textit{Bug}\xspace}
\newcommand{\RUser}{\textit{User}\xspace}
\newcommand{\REmail}{\textit{Email}\xspace}
\newcommand{\RRel}{\textit{Rel}\xspace}
\newcommand{\pToI}{\textit{pers2iri}\xspace}
\newcommand{\bToI}{\textit{bug2iri}\xspace}

\newcommand{\rdffrominstance}{\operatorname{\mathit{inst-to-rdf}}}
\newcommand{\rdftoinst}{\operatorname{\mathit{rdf-to-inst}}}

\begin{document}
\maketitle
\begin{abstract}
  We study the relational to RDF data exchange problem, where the
  target constraints are specified using Shape Expression schema
  (ShEx). We investigate two fundamental problems: 1)
  \emph{consistency} which is checking for a given data exchange
  setting whether there always exists a solution for any source
  instance, and 2) constructing a \emph{universal solution} which is a
  solution that represents the space of all solutions. We propose to
  use \emph{typed IRI constructors} in source-to-target tuple generating dependencies
  to create the IRIs of the RDF graph from the values in the relational instance,
  and we translate ShEx into a set of target 
  dependencies. We also identify data exchange settings that are
  \emph{key covered}, a property that is decidable and guarantees
  consistency. Furthermore, we show that this property is a sufficient
  and necessary condition for the existence of universal solutions for
  a practical subclass of \emph{weakly-recursive} ShEx.
\end{abstract}


\vspace{-2em}
\section{Introduction}
\vspace{-1em}
\emph{Data exchange} can be seen as a process of transforming an instance of one schema, called the \emph{source schema}, to an instance of another schema, called the \emph{target schema}, according to a set of rules, called \emph{source-to-target tuple generating dependencies} (st-tgds). But more generally, for a given source schema, any instance of the target schema that satisfies the dependencies is a \emph{solution} to the data exchange problem. Naturally, there might be no solution, and then we say that the setting is \emph{inconsistent}. Conversely, there might be a possibly infinite number of solutions, and a considerable amount of work has been focused on finding a \emph{universal solution}, which is an instance (potentially with incomplete information) that represents the entire space of solutions. Another fundamental and well-studied problem is checking \emph{consistency} of a data exchange setting i.e., given the source and target schemas and the st-tgds, does a solution exist for any source instance. For relational databases the consistency problem is in general known to be undecidable ~\cite{kolaitis:2006,beeri:1981} but a number of decidable and even tractable cases has been identified, for instance when a set of weakly-acyclic dependencies is used~\cite{fagin:2005a}.

\emph{Resource Description Framework} (RDF) \cite{w3c:rdf} is a well-established  format for publishing linked data on the Web, where \emph{triples} of the form $(\mathit{subject}, \mathit{predicate}, \mathit{object})$ allow to represent an edge-labeled graph.  While originally RDF was introduced schema-free to promote its adoption and wide-spread use, the use of RDF for storing and exchanging data among web applications has prompted the development of schema languages for RDF~\cite{ryman:2013,sirin:2010,w3c:shacl}. One such schema language, under continuous development, is Shape Expressions Schemas (ShEx)~\cite{staworko:2015a,boneva:17a}, which allows to define structural constraints on nodes and their immediate neighborhoods in a declarative fashion.

In the present paper, we study the problem of data exchange where the source is a relational database and the target is an RDF graph constrained with a ShEx schema. Although an RDF graph can be seen as a relational database with a single ternary relation $\Triple$, RDF graphs require using  \emph{Internationalized Resource Identifiers} (IRIs) as global identifiers for entities. Consequently, the framework for data exchange for relational databases cannot be directly applied \emph{as is} and we adapt it with the help of \emph{IRI constructors}, functions that assign IRIs to identifiers from a relational database instance. Their precise implementation is out of the scope of this paper and belongs to the vast domain of entity matching~\cite{kopcke:2010}. 

\begin{example}\label{ex:data-exchange}
Consider the relational database of bug reports in Figure~\ref{fig:db}, where the relation $\RBug$ stores a list of bugs with their description and ID of the user who reported the bug, the name of each user is stored in the relation $\RUser$ and her email in the relation $\REmail$. Additionally, the relation $\RRel$ identifies related bug reports for any bug report. 
\begin{figure}[htb]
\centering
\vspace{-1em}
\begin{tabular}[t]{c|ccc}
    \RBug \hspace{2pt}&\hspace{2pt}
    \uline{\textit{bid}}& \textit{descr} & \textit{uid}\\[2pt]
    \hline
    &1 & Boom! & 1 \\
    &2 & Kaboom! & 2 \\
    &3 & Kabang! & 1 \\
    &4 & Bang! & 3 
  \end{tabular}
  \hfill
  \begin{tabular}[t]{c|cc}
    \RUser \hspace{2pt}&\hspace{2pt}
    \uline{\textit{uid}}& \textit{name} \\[2pt]
    \hline
    & 1 & Jose\\
    & 2 & Edith\\
    & 3 & Steve89
  \end{tabular}
  \hfill
  \begin{tabular}[t]{c|cc}
    \REmail \hspace{2pt}&\hspace{2pt}
    \uline{\textit{uid}}& \textit{email} \\[2pt]
    \hline
    & 1 & j@ex.com\\
    & 2 & e@o.fr
  \end{tabular}
  \hfill
  \begin{tabular}[t]{c|cc}
    \RRel \hspace{2pt}&\hspace{2pt}
    \uline{\textit{bid}}& \uline{\textit{rid}} \\[2pt]
    \hline
    & 1 & 3 \\
    & 1 & 4\\
    & 2 & 4 
  \end{tabular}
\caption{Relational database (source)\label{fig:db}}
\end{figure}
\vspace{-10pt}

\noindent
Now, suppose that we wish to share the above data with a partner that has an already existing infrastructure for consuming bug reports in the form of RDF whose structure is described with the following ShEx schema (where $\col$ is some default prefix):
\begin{align*}
  \Bug \to{}& \{ \descr \dbl \LitSym^\ONE, \rep \dbl \User^\ONE, \rel \dbl \Bug^\MANY\}\\
  \User \to{}& \{ \name\dbl \LitSym^\ONE, \mel\dbl \LitSym^\ONE, \phone \dbl \LitSym^\MAYBE \}
\end{align*}
The above schema defines two types of (non-literal) nodes: $\Bug$ for describing bugs and $\User$ for describing users. Every bug has a description, a user who reported it, and a number of related bugs. Every user has a name, an email, and an optional phone number. 
The reserved symbol $\LitSym$ indicates that the corresponding value is a literal.

The mapping of the contents of the relational database to RDF is defined with the following logical rules (the free variables are implicitly universally quantified).
\begin{align*}
  \RBug(b,d,u) \Rightarrow{}& \Triple(\bToI(b),\descr,d) \land \Bug(\bToI(b)) \land{}  \\
  &\Triple(\bToI(b),\rep,\pToI(u))\\
  \RRel(b_1,b_2) \Rightarrow{}& \Triple(\bToI(b_1),\rel,\bToI(b_2))\\
  \RUser(u,n)  \Rightarrow{}& \Triple(\pToI(u),\name,n) \land \User(\pToI(u)) \\
  \RUser(u,n) \land \REmail (u,e) \Rightarrow{}& \Triple(\pToI(u), \mel, e) \land \LitSym(e)
\end{align*}
On the left-hand-side of each rule we employ queries over the source relational database, while on the right-hand-side we make corresponding assertions about the triples in the target RDF graph and the types of the nodes connected by the triples. The atomic values used in relational tables need to be carefully converted to IRIs with the help of IRI constructors $\pToI$ and $\bToI$. The constructors can be \emph{typed} i.e., the IRI they introduce are assigned a unique type in the same st-tgd.

We point out that in general, IRI constructors may use external data sources to properly assign to the identifiers from the relational database unique IRIs that identify the object in the RDF domain. For instance, the user Jose is our employee and is assigned the corresponding IRI $\mathsf{emp}\mathord{:}\mathsf{jose}$, the user Edith is not an employee but a registered user of our bug reporting tool and consequently is assigned the IRI $\mathsf{user}\mathord{:}\mathsf{edith}$, and finally, the user Steve89 is an anonymous user and is assigned a special IRI indicating it $\mathsf{anon\col 3}$.  

Figure~\ref{fig:rdf} presents an RDF instance that is a solution to the problem at hand.
We point out that the instance uses a (labeled) null literal $\bot_1$ for the email of Steve89 that is required by the ShEx schema but is missing in our database.\qed
\begin{figure}[htb]
\centering 
\vspace{-1em}
\newcommand{\sequal}{\hspace{1pt}\mathord{=}\hspace{1pt}}
\begin{tikzpicture}[>=latex]  
  \node (d3) at (-4,2.3) { ``Kabang!''};
  \node (d1) at (-1.5,2.5) { ``Boom!''};
  \node (d4) at (1.15,2.5) { ``Bang!''};
  \node (d2) at (4,2.5) { ``Kaboom!''};
  
  \node (bug3) at (-4,1) {$\mathsf{bug\col 3}$};
  \node (bug1) at (-1.5,1.25) {$\mathsf{bug\col 1}$};
  \node (bug4) at (1.15,1) {$\mathsf{bug\col 4}$};
  \node (bug2) at (4,1.25) {$\mathsf{bug\col 2}$};

  \node (user1) at (-2.5,0) {$\mathsf{emp}\col\mathsf{jose}$};
  \node (user2) at (3.75,0) {$\mathsf{user}\col\mathsf{edith}$};
  
  \node (emp1) at (0.9,-0.25) {$\mathsf{anon\col 3}$};  
  
  \node (n1) at (-3.25,-1.5) { ``Jose''};
  \node (e1) at (-1.75,-1.5) { ``j@ex.com''};

  \node (n2) at (3,-1.5) { ``Edith''};
  \node (e2) at (4.5,-1.5) { ``e@o.fr''};

  \node (n3) at (0.1,-1.65) { ``Steve89''};
  \node (e3) at (1.8,-1.65) {$\bot_1$};
  
  \draw (user1) edge[->] node[above,sloped] {\name} (n1);
  \draw (user1) edge[->] node[above,sloped] {\mel} (e1);
  
  \draw (user2) edge[->] node[above,sloped] {\name} (n2);
  
  \draw (emp1) edge[->] node[above,sloped] {\name} (n3);
  \draw (emp1) edge[->] node[above,sloped] {\mel} (e3);

  \draw (user2) edge[->] node[above,sloped] {\mel} (e2);    
  
  \draw (bug1) edge[->] node[above,sloped] {\rel} (bug4);
  \draw (bug1) edge[->] node[above,sloped] {\rel} (bug3);
  \draw (bug1) edge[->] node[above,sloped] {\rep} (user1);
  \draw (bug1) edge[->] node[above,sloped] {\descr} (d1);
  
  \draw (bug2) edge[->] node[above,sloped] {\rel} (bug4);
  \draw (bug2) edge[->] node[above,sloped] {\rep} (user2);
  \draw (bug2) edge[->] node[above,sloped] {\descr} (d2);
  
  \draw (bug3) edge[->] node[above,sloped] {\rep} (user1);
  \draw (bug3) edge[->] node[above,sloped] {\descr} (d3);
  
  \draw (bug4) edge[->] node[below,sloped] {\rep} (emp1);
  \draw (bug4) edge[->] node[above,sloped] {\descr} (d4);
\end{tikzpicture}
\caption{Target RDF graph (solution)\label{fig:rdf}}
\end{figure}
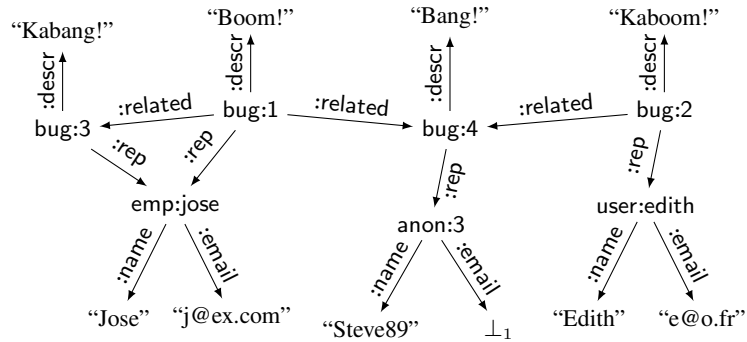
\end{example}
\vspace{-10pt}
The presence of target schema raises the question of consistency.
On the one hand, we can prove that for any instance of the relational database in Example~\ref{ex:data-exchange} there exists a target solution that satisfies the schema and the set of source-to-target tuple generating dependencies. On the other hand, suppose we 
allow
a user to have multiple email addresses, by changing the \emph{key} of $\REmail$ to both $\mathit{uid}$ and $\mathit{email}$). Then, the setting would not be consistent as one could construct an instance of the relational database, with multiple email addresses for a single user, for which there would be no solution.

Our investigation provides a preliminary analysis of the consistency problem for relational to RDF data exchange with target ShEx schema. Our contribution can be summarized as follows:
\begin{itemize}
\item a formalization of relational to RDF data exchange with target ShEx schema and typed IRI constructors.
\item a decidable characterization of a \emph{fully-typed} \emph{key-covered} data exchange setting that is a sufficient and necessary condition for consistency. 
\item an additional restriction of weak-recursion on ShEx schemas that ensures the existence of universal solution. 
\end{itemize}

\noindent
\textbf{Related Work.} \textit{Relational Data Exchange, Consistency.}
The theoretical foundations of data exchange for relational databases are laid in \cite{fagin:2005a,arenas:2010a}. 
Source-to-target dependencies with Skolem functions were introduced by nested dependencies \cite{DBLP:conf/vldb/FuxmanHHMPP06} in order to improve the quality of the data exchange solution.
General existentially quantified functions are possible in second order tgds \cite{arenas:2013a}.
Consistency in the case of relational data exchange is undecidable, and decidable classes usually rely on chase termination  ensured by restrictions such as acyclicity, or guarded dependencies, or restrictions on the structure of source instances.
The consistency criterion that we identify in this paper is orthogonal and is particular to the kind of target constraints imposed by ShEx schemas.
In \cite{marnette:2010}, static analysis is used to test whether a target dependency is implied by a data exchange setting, these however rely on chase termination.
Consistency is an important problem in XML data exchange \cite{arenas:2010a} but the techniques developed for XML do not apply here.

\smallskip
\noindent
\textit{Value Invention, Relational to RDF Data Exchange.}
Value invention is used in the purely relational setting for generating null values.
Tools such as Clio \cite{fagin:2009} and ++Spicy \cite{marnette:2011} implement Skolem functions as concatenation of their arguments.
IRI value invention is considered by R2RML \cite{das:2011}, a W3C standard for writing customizable relational to RDF mappings.
The principle is similar to what we propose here.
A R2RML mapping allows to specify logical tables (i.e. very similar to left-hand-sides of source-to-target dependencies), and then how each row of a logical table is used to produce one or several triples of the resulting RDF graph.
Generating IRI values in the resulting graph is done using templates that specify how a fixed IRI part is to be concatenated with the values of some of the columns of the logical table.
R2RML does not allow to specify structural constraints on the resulting graph, therefore the problem of consistency is irrelevant there.
In \cite{sequeda:2012a}, a direct mapping that is a default automatic way for translating a relational database to RDF is presented.
The main difference with our proposal and with R2RML is that the structure of the resulting RDF graph is not customizable.
In \cite{boneva:2015b} we studied relational to graph data exchange in which the target instance is an edge labelled graph and source-to-target and target dependencies are conjunctions of nested regular expressions.
Such a framework raises a different kind of issues, among which is the materialization of a solution, as a universal solution is not necessarily a graph itself, but a graph pattern in which some edges carry regular expressions.
On the other hand, IRI value invention is not relevant in such framework.

\medskip
\noindent
\textbf{Organization.} In Section~\ref{sec:prelim} we present basic notions. In Section~\ref{sec:shex-fo} we show how ShEx schemas can be encoded using target dependencies. In Section \ref{sec:data-exchange} we formalize relational to RDF data exchange. In Section \ref{sec:consistency-des} we study the problem of consistency. And finally, in Section~\ref{sec:universal-solution} we investigate the existence of universal solutions. Conclusions and directions of future work are in Section~\ref{sec:concl-future-work}. 

\vspace{-1em}
\section{Preliminaries}
\label{sec:prelim}
\vspace{-1em}

\subsubsection*{First-order logic.}
A \emph{relational signature} $\Rsig$ (resp. \emph{functional signature} $\Fsig$) is a finite set of relational symbols (resp. functional symbols), each with fixed arity. A \emph{type symbol} is a relational symbol with arity one. A \emph{signature} is a set of functional and relational symbols. In the sequel we use $\Rsig$, resp. $\Fsig$, resp. $\Tsig$ for sets of relational, resp. functional, resp. type symbols.

We fix an infinite and enumerable domain $\Dom$ partitioned into three infinite subsets $\Dom = \IriSet\cup \LitSet \cup \BlankSet$ of IRIs, literals, and blank nodes respectively. Also, we assume an infinite subset $\NullLitSet \subseteq \LitSet$ of null literals. In general, by \emph{null} values we understand both null literals and blank nodes and we denote them by $\NullSet=\NullLitSet\cup\BlankSet$.

Given a signature $\FOsig = \Rsig \cup \Fsig$, a \emph{model} (or a \emph{structure}) of $\FOsig$ is a mapping $\Model$ that with any symbol $S$ in $\FOsig$ associates its interpretation $S^\Model$ s.t.:
\begin{itemize}
\item $R^\Model \subseteq \Dom^n$ for any relational symbol $R \in \Rsig$ of arity $n$;
\item $f^\Model: \Dom^n \to \Dom$, which is a total function for any function symbol $f \in \Fsig$ of arity $n$.
\end{itemize}

We fix a countable set $\Vars$ of variables and reserve the symbols $x, y, z$ for variables, and the symbols $\vect{x}$, $\vect{y}$, $\vect{z}$ for vectors of variables.
We assume that the reader is familiar with the syntax of first-order logic with equality and here only recall some basic notions. 
A \emph{term} over $\Fsig$ is either a variable in $\Vars$, or a constant in $\Dom$, or is of the form $f(\vect{x})$ where $f \in \Fsig$ and the length of $\vect{x}$ is equal to the arity of $f$; we remark that we do not allow nesting of function symbols in terms.
A \emph{dependency} is a formula of the form $\forall\vect{x}.\varphi\Rightarrow\exists\vect{y}.\psi$ and in the sequel, we often drop the universal quantifier, write simply $\varphi\Rightarrow\exists\vect{y}.\psi$, and assume that implicitly all free variables are universally quantified. 

The semantics of first-order logic formulas is captured with the \emph{entailment} relation $\Model,\nu \models \phi$ defined in the standard fashion for a model $\Model$, a first-order logic formula $\phi$ with free variables $\vect{x}$ and a valuation $\nu: \vect{x} \to \Dom$. The entailment relation is extended to sets of formulas in the canonical fashion: $\Model \models \{\varphi_1,\ldots,\varphi_n\}$ iff $\Model \models \varphi_i$ for every $i\in\{1,\ldots,k\}$.

\vspace{-1em}
\subsubsection*{Relational Databases.}
We model relational databases using relational structures in the standard fashion. For our purposes we are only concerned with functional dependencies, which include key constraints. Other types of constraints, such as inclusion dependencies and foreign key constraints, are omitted in our abstraction. 

A \emph{relational schema} is a pair $\Rsch = (\Rsig, \fdep)$ where $\Rsig$ is a relational signature and $\fdep$ is a set of \emph{functional dependencies} (fds) of the form $R:X\rightarrow Y$, where $R\in\Rsig$ is a relational symbol of arity $n$, and $X,Y\subseteq\{1,\ldots,k\}$. An fd $R:X\rightarrow Y$ is a short for the following formula
$\forall\vect{x},\vect{y}.\ R(\vect{x})\land R(\vect{y})\land \textstyle\bigwedge_{i\in X} (x_i=y_i) \Rightarrow \textstyle\bigwedge_{j\in Y} (x_j=y_j)$.
An \emph{instance} of $\Rsch$ is a model $\Inst$ of $\Rsig$ and we
say that $\Inst$ is \emph{valid} if $\Inst\models\fdep$. 
The \emph{active domain} $\dom(\Inst)$ of the instance $\Inst$ is the set of values from $\Dom$ that appear in $R^\Inst$ for some relational symbol $R$ in $\Rsig$. Unless we state otherwise, in the sequel we consider only instances that use only constants from $\LitSet\setminus\NullLitSet$.  

\vspace{-1em}
\subsubsection*{RDF Graphs and Shape Expressions Schemas.}
Recall that an \emph{RDF graph}, or \emph{graph} for short, is a set of triples in $(\IriSet \cup \BlankSet) \times \IriSet \times (\IriSet \cup \BlankSet \cup \LitSet)$.
The set of \emph{nodes} of the graph $G$ is the set of elements of $\IriSet \cup \BlankSet \cup \LitSet$ that appear on first or third position of a triple in $G$.

We next define the fragment of shape expression schemas that we consider, and that was called  $\mathsf{RBE_0}$ in \cite{}. Essentially, a ShEx is a collection of shape names, and each comes with a  definition consisting of a set of triple constraints. A triple constraint indicates a label of an outgoing edge, the shape of the nodes reachable with this label, and a multiplicity indicating how many instances of this kind of edge are allowed. We remark that the constraints expressible with this fragment of ShEx, if non-recursive, can also be captured by a simple fragment of SHACL with \textsc{and} operator only.

Formally, a \emph{multiplicity} is an element of $\{\ONE, \MAYBE, \MANY, \PLUS\}$ with the natural interpretation: $\ONE$ is exactly one occurrence, $\MAYBE$ stands for none or one occurrence, $\MANY$ stands for an arbitrary number of occurrences, and $\PLUS$ stands for a positive number of occurrences. A \emph{triple constraint} over a finite set of shape names $\Tsig$ is an element of $\IriSet \times (\Tsig \cup \{\LitSym\}) \times \{\ONE, \MAYBE, \MANY, \PLUS\}$, where $\LitSym$ is an additional symbol used to indicate that a node is to be a literal. Typically, we shall write a triple constraint $(p, T, \mu)$ as $p\dbl T^\mu$. Now, a \emph{shape expressions schema}, or \emph{ShEx} schema for short, is a couple $\ShexSch = (\Tsig, \delta)$ where $\Tsig$ is a finite set of shape names, and $\delta$ is shape definition function that maps every symbol $T\in\Tsig$ to a finite set of triple constraints over $\Tsig$ such that for every shape name $T$ and for every IRI $p$, $\delta(T)$ contains at most one triple constraint using $p$.

For a finite set $\Tsig$ of shape names, a \emph{$\Tsig$-typed} graph is a couple $(G, \typing)$ where $G$ is a graph and $\typing$ is a mapping from the nodes of $G$ into $2^{\Tsig \cup \{\LitSym\}}$ that with every node of $G$ associates a (possibly empty) set of types.
Let $\ShexSch = (\Tsig, \delta)$ be a ShEx schema.
The $\Tsig$-typed graph $(G, \typing)$ is \emph{correctly typed} w.r.t.\ $\ShexSch$ if it satisfies the constraints defined by $\delta$ i.e., for any node $n$ of $G$:
\begin{itemize}
\item if $\LitSym \in \typing(n)$, then $n \in \LitSet$;
\item if $T \in \typing(n)$ then $n \in \IriSet$ and for every  $p\dbl S^\mu$ in $\delta(T)$ we have that (1) for any triple $(n, p, m)$ in $G$, $S$ belongs to $\typing(m)$, and (2) if $K$ is the set of triples in $G$ whose first element is $n$ and second element is $p$, then the cardinality of $K$ is bounded by $\mu$ i.e., $|K|=1$ if $\mu= \ONE$, $|K|\leq 1$ if $\mu = \MAYBE$, and $|K|\geq 1$ if $\mu = \PLUS$ (there is no constraint if $\mu=\MANY$).
\end{itemize}
For instance, a correct typing for the graph in Figure~\ref{fig:rdf} assigns the type $\Bug$ to the nodes $\mathsf{bug\col 1}$, $\mathsf{bug\col 2}$, $\mathsf{bug\col 3}$, and $\mathsf{bug\col 4}$; the type $\User$ to the nodes $\mathsf{emp}\col\mathsf{jose}$, $\mathsf{user}\col\mathsf{edith}$, and $\mathsf{anon}\col\mathsf{3}$; and $\LitSym$ to every literal node.




\vspace{-1em}
\section{ShEx Schemas as Sets of Dependencies}
\label{sec:shex-fo}
\vspace{-1em}
In this section we show how to express a ShEx schema $\ShexSch = (\Tsig, \delta)$ using dependencies. 

First, we observe that any $\Tsig$-typed graph can be easily converted to a relational structure over the relational signature $\GTsig = \{\Triple\} \cup \Tsig \cup\{\LitSym\}$, where $\Triple$ is a ternary relation symbol for encoding triples, and $\Tsig \cup\{\LitSym\}$ are monadic relation symbols indicating node types (details in Appendix~\ref{app:shex-fo}). 
Consequently, in the sequel, we may view a $\Tsig$-typed graph as the corresponding relational structure (or even a relational database over the schema $(\GTsig, \emptyset)$).

Next, we define auxiliary dependencies for any two $T,S\in\Tsig$ and any $p\in\IriSet$ 
\begin{align*}
  \tconstr(T, S, p) \colonequals{} & T(x) \land \Triple(x, p, y) \Rightarrow S(y)\\
  \multleastone(T, p) \colonequals{} & T(x) \Rightarrow \exists y. \Triple(x, p, y)\\
  \multmostone(T, p) \colonequals{} & T(x) \land \Triple(x, p, y) \land \Triple(x, p, z) \Rightarrow y = z
\end{align*}
We point out that in terms of the classical relational data exchange, $\tconstr$ and $\multleastone$ are \emph{tuple generating dependencies} (\emph{tgd}s), and $\multmostone$ is an \emph{equality generating dependency} (\emph{egd}). We capture the ShEx schema $\ShexSch$ with the following set of dependencies:
\begin{align*}
\shexdep = {} 
  &\{\tconstr(T,S,p)\mid T\in\Tsig,\ p\dbl{}S^\mu\in\delta(T)\}\cup{}\\
  &\{\multleastone(T, p) \mid T\in\Tsig,\ p\dbl{}S^\mu\in\delta(T),\ \mu\in\{\ONE,\PLUS\} \}\cup{}\\
  &\{\multmostone(T, p) \mid T\in\Tsig,\ p\dbl{}S^\mu\in\delta(T),\  \mu\in\{\ONE,\MAYBE\} \}.
\end{align*}
\begin{lemma}
  \label{lem:shex-as-dependencies}
  For every ShEx schema $\ShexSch = (\Tsig, \delta)$ and every $\Tsig$-typed RDF graph $(G, \typing)$,
  $(G, \typing)$ is correctly typed w.r.t.\ $\ShexSch$ iff $(G, \typing) \models \shexdep$.
\end{lemma}


\vspace{-1em}
\section{Relational to RDF Data Exchange}
\label{sec:data-exchange}
\vspace{-1em}
In this section, we present the main definitions for data exchange.

\vspace{-0.5em}
\begin{definition}[Data exchange setting]
A \emph{relational to RDF data exchange setting} is a 
tuple $\DES = (\Rsch, \ShexSch, \stdep, \Fsig, \Fint)$ where 
$\Rsch = (\Rsig, \fdep)$ is a source relational schema, 
$\ShexSch = (\Tsig, \delta)$ is a target ShEx schema, 
$\Fsig$ is a function signature, 
$\Fint$ as an interpretation for $\Fsig$ that with every function symbol $f$ in $\Fsig$ of arity $n$ associates a function from $\Dom^n$ to $\IriSet$, and
$\stdep$ is a set of \emph{source-to-target tuple generating dependencies}, clauses of the form $\forall\vect{x}. \varphi \Rightarrow \psi$, where $\varphi$ is a conjunction of atomic formulas over the source signature $\Rsig$ and $\psi$ is a conjunction of atomic formulas over the target signature 
$\GTsig \cup \Fsig$. Furthermore, we assume that all functions in $\Fint$ have disjoint ranges i.e., for $f_1,f_2\in\Fint$ if $f_1\neq f_2$, then $\ran(f_1)\cap\ran(f_2)=\emptyset$.
\end{definition}

\vspace{-1em}
\begin{definition}[Solution]
Take a data exchange setting $\DES = (\Rsch, \ShexSch, \stdep, \Fsig, \Fint)$, and let $\Inst$ be a valid instance of $\Rsch$. Then, a \emph{solution} for $\Inst$ w.r.t.\  $\DES$ is any 
$\Tsig$-typed graph $\Ginst$ such that $\Inst \cup \Ginst \cup \Fint \models \stdep$ and $\Ginst \models \shexdep$.
\end{definition}

\vspace{-0.5em}
A \emph{homomorphism} $h: \Inst_1 \to \Inst_2$ between two relational structures $\Inst_1, \Inst_2$ of the same relational signature $\Rsig$ is a mapping from $\dom(\Inst_1)$ to $\dom(\Inst_2)$ that 1) preserves the values of non-null elements i.e., $h(a)=a$ whenever $a\in\dom(\Inst_1)\setminus\NullSet$, and 2) for every $R\in\Rsig$ and every $\vect{a}\in R^{I_1}$ we have $h(\vect{a})\in R^{I_2}$, where $h(\vect{a})=(h(a_1),\ldots,h(a_n))$ and $n$ is the arity of $R$.

\vspace{-0.5em}
\begin{definition}[Universal Solution]
  Given a data exchange setting $\DES$ and a valid source instance $\Inst$, a solution $\Ginst$ for $\Inst$ w.r.t.\ $\DES$ is \emph{universal}, if for any solution $J'$ for $\Inst$ w.r.t.\ $\DES$ there exists a homomorphism $h: J\to J'$.
\end{definition}

\vspace{-0.5em}
As usual, a solution is computed using the chase.
We use a slight extension of the standard chase (explained in the appendix) in order to handle function terms, which in our case is simple (compared to e.g. \cite{arenas:2013a}) as the interpretation of function symbols is given.


\vspace{-1em}
\section{Consistency}
\label{sec:consistency-des}
\vspace{-1em}

\begin{definition}[Consistency]
  A data exchange setting $\DES$ is \emph{consistent} if every valid source instance admits a solution. 
\end{definition}

We fix a relational to RDF data exchange setting $\DES = (\Rsch, \ShexSch, \stdep, \Fsig, \Fint)$ and let $\ShexSch = (\Tsig, \delta)$. We normalize source-to-target tuple generating dependencies so that their right-hand-sides use exactly one $\Triple$ atom and at most two type assertions on the subject and the object of the triple; such normalization is possible as our st-tgds do not use existential quantification. In this paper, we restrict our investigation to completely typed st-tgds having both type assertions, and therefore being of the following form 
\[
\forall \vect{x}.\ \varphi \Rightarrow \Triple(s,p,o) \land T_s(s) \land T_o(o),
\]
where $s$ is the \emph{subject term}, $T_s$ is the \emph{subject type}, $p\in\IriSet$ is the \emph{predicate}, $o$ is the \emph{object term}, and $T_o$ is the \emph{object type}. Because the subject of a triple cannot be a literal, we assume that $s=f(\vect{y})$ for $f\in\Fsig$ and for $\vect{y}\subseteq\vect{x}$, and $T_s\in\Tsig$. As for the object, we have two cases: 1) the object is an IRI and then $o=g(\vect{z})$ for $g\in\Fsig$ and for $\vect{z}\subseteq\vect{x}$, and $T_o\in\Tsig$, or 2) the object is literal $o=z$ for $z\in\vect{x}$ and $T_o=\LitSym$. 
Moreover, we assume consistency with the target ShEx schema $\ShexSch$ i.e., for any st-tgd in $\stdep$ with source type $T_s$, predicate $p$, and object type $T_o$ we have $p\dbl{}T_o^\mu\in\delta(T_s)$ for some multiplicity $\mu$. Finally, we assume that every IRI constructor in $\Fsig$ is used with a unique type in $\Tsig$. When all these assumptions are satisfied, we say that the source-to-target tuple generating dependencies are \emph{fully-typed}.

While the st-tgds in Example~\ref{ex:data-exchange} are not fully-typed, an equivalent set of fully-typed dependencies can be easily produced if additionally appropriate foreign keys are given. For instance, assuming the foreign key constraint $\RBug[\mathit{uid}]\subseteq\RUser[\mathit{uid}]$, the first rule with $\RBug$ on the left-hand-side is equivalent to 
\begin{small}
\begin{align*}
  &\RBug(b,d,u)  \Rightarrow\Triple(\bToI(b),\descr,d) \land \Bug(\bToI(b)) \land \LitSym(d)\\
  &\RBug(b,d,u)  \Rightarrow \Triple(\bToI(b),\rep,\pToI(u)) \land \Bug(\bToI(b)) \land \User(\pToI(u)) 
\end{align*}%
\end{small}%

Now, two st-tgds are \textit{contentious} if both use the same IRI constructor $f$ for their subjects and have the same predicate, hence the same subject type $T_s$ and object type $T_o$, and $p\dbl{}T_o^\mu\in\delta(T_s)$ with $\mu=\ONE$ or $\mu=\MAYBE$. We do not want two contentious st-tgds to produce two triples with the same subject and different objects. Formally, take two contentious st-tgds $\sigma_1$ and $\sigma_2$ and assume they have the form (for $i\in\{1,2\}$, and assuming $\vect{x}_1, \vect{x}_2, \vect{y}_1,\vect{y}_2$ are pairwise disjoint)
\[
\sigma_i = \forall \vect{x}_i,\vect{y}_i.\ \varphi_i(\vect{x}_i,\vect{y}_i) \Rightarrow \Triple(f(\vect{x}_i),p,o_i) \land T_s(f(\vect{x}_i)) \land T_o(o_i).
\]
The st-tgds $\sigma_1$ and $\sigma_2$ are \emph{functionally overlapping} if for every valid instance $I$ of $\Rsch$
\[
I\cup \Fint\models \forall \vect{x}_1,\vect{y}_1,\vect{x}_2,\vect{y}_2.\ 
\varphi_1(\vect{x}_1,\vect{y}_1)\land
\varphi_2(\vect{x}_2,\vect{y}_2)\land
\vect{x}_1=\vect{x}_2\Rightarrow 
o_1=o_2.
\]
Finally, a data-exchange setting is \emph{key-covered} if every pair of its contentious st-tgds is functionally overlapping. Note that any single st-tgd may be contentious with itself. 

\begin{theorem}
  \label{thm:key-covered-fully-typed-iff-consistent} 
  A fully-typed data exchange setting is consistent if and only if it is key-covered.
\end{theorem}
The sole reason for the non-existence of a solution for a source instance $\Inst$ is a violation of some egd in $\shexdep$. The key-covered property ensures that such egd would never be applicable. Intuitively, two egd-conflicting objects $o_1$ and $o_2$ are necessarily generated by two contentious st-tgds. The functional-overlapping criterion guarantees that the terms $o_1$ and $o_2$ are ``guarded'' by a primary key in the source schema, thus cannot be different.
\vspace{-0.5em}
\begin{theorem}
  \label{thm:key-covered-decidable}
  It is decidable whether a fully-typed data exchange setting is key-covered.
\end{theorem}
The proof uses a reduction to the problem of functional dependency propagation \cite{klug:1982a}.


\vspace{-1em}
\section{Universal Solution}
\label{sec:universal-solution}
\vspace{-1em}

In this section, we identify conditions that guarantee the existence of a universal solution. Our results rely on the existence of a universal solution for sets of weakly-acyclic sets of dependencies for relational data exchange~\cite{fagin:2005a}. As the tgds and egds that we generate are driven by the schema (cf. Section~\ref{sec:shex-fo}), we introduce a restriction on the ShEx schema that yields weakly-acyclic sets of dependencies, and consequently, guarantees the existence of universal solution.

The \emph{dependency graph} of a ShEx schema $\ShexSch = (\Tsig, \delta)$ is the directed graph whose set of nodes is $\Tsig$ and has an edge $(T,T')$ if $T'$ appears in some triple constraint $p::T'^\mu$ of  $\delta(T)$. There are two kinds of edges: \emph{strong edge}, when the multiplicity $\mu \in \{\ONE,\PLUS\}$, and \emph{weak edge}, when $\mu \in \{\MANY,\MAYBE\}$. The schema $\ShexSch$ is \emph{strongly-recursive} if its dependency graph contains a cycle of strong edges only, and is \emph{weakly-recursive} otherwise. Take for instance the following extension of the ShEx schema from Example~\ref{ex:data-exchange}:
\begin{align*}
  \User \to{}& \{ \name\dbl \LitSym^\ONE, \mel\dbl \LitSym^\ONE , \phone \dbl \LitSym^\MAYBE\}\\
  \Bug \to{}& \{ \rep \dbl \User^\ONE, \descr \dbl \LitSym^\ONE, \rel \dbl \Bug^\MANY, \repr \dbl \Emp^\MAYBE\}\\
  \Emp \to{}& \{ \name\dbl \LitSym^\ONE, \prep\dbl \Test^\PLUS \}\\
  \Test \to{}& \{ \grp \dbl \Bug^\PLUS\}
\end{align*}
The dependency graph of this schema, presented in Figure~\ref{fig:dg}. contains two cycles but neither of them is strong. Consequently, the schema is weakly-recursive (and naturally so is the ShEx schema in Example~\ref{ex:data-exchange}).
\begin{figure}[t]
	\centering
        \vspace{-0.5cm}
	\begin{tikzpicture}[>=latex]
	\node (Bug) at (-6.5, 0.5) {\Bug};
	\node (User)[left=1cm of Bug] {\User};
	\node (Emp)[right=1.9cm of Bug]  {\Emp};
        \node (Ta)[above right=0.7cm of Bug] {\Test};
	
	\draw[loop,bend angle=45]  (Bug) edge[->,dashed] (Bug); 
	
	\draw (Emp) [bend angle=45] edge[->] (Ta); 
	\draw (Ta) [bend angle=45] edge[->] (Bug); 
	
	\draw (Bug) edge[->]
	(User);
	\draw (Bug) edge[->,dashed] 
	(Emp);
	
	
	
	
      \end{tikzpicture}
      \caption{Dependency graph with dashed weak edges and plain strong edges \label{fig:dg}\vspace{-2em}}
\end{figure}
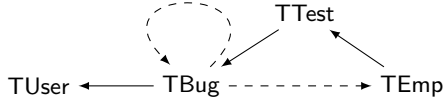

As stated above, a weakly-recursive ShEx schema guarantees a weakly-acyclic set of dependencies and using results from~\cite{fagin:2005a} we get
\begin{proposition}
\label{prop:universal-solution-weakly-recursive}
   Let $\DES = (\Rsch, \ShexSch, \stdep, \Fsig, \Fint)$ be a data exchange setting and $\Inst$ be a valid instance of $\Rsch$.
   If $\ShexSch$ is weakly recursive, then every chase sequence of $\Inst$ with $\stdep \cup \shexdep$ is finite, and either every chase sequence of $\Inst$ with $\stdep$ fails, or every such chase sequence computes a universal solution of $\Inst$ for $\DES$.
\end{proposition}


\vspace{-2em}
\section{Conclusion and Future Work}
\label{sec:concl-future-work}
\vspace{-1em}
We presented a preliminary study of the consistency problem for relational to RDF data exchange in which the target schema is ShEx.
Consistency is achieved by fully-typed and key-covered syntactic restriction of st-tgds.
An open problem that we plan to investigate is consistency when the fully-typed restriction is relaxed; we believe that it is achievable if we extend the definition of contentious st-tgds.
Another direction of research is to consider a larger subset of ShEx.
Finally, we plan to extend our framework to typed literals which are not expected to bring fundamental difficulties but are essential for practical applications.
\vspace{-0.5em}

\bibliographystyle{splncs03}  
\bibliography{exchange}

\newpage
\appendix
\section{ShEx Schemas as Sets of Dependencies}
\label{app:shex-fo}

 \begin{lemma}
\label{lem:one-to-one-correspondance}
For any $\Tsig$-typed graph $(G, \typing)$, let $\rdftoinst(G,\typing)$ be defined as below.
For any $\Inst$ instance of $(\GTsig, \emptyset)$ satisfying  $\Triple^\Inst \subseteq (\IriSet \cup \BlankSet) \times \IriSet \times (\IriSet \cup \BlankSet \cup \LitSet)$ and $\LitSym^\Inst \subseteq \LitSet$ and $T^\Inst \subseteq \IriSet$ for all $T \in \Tsig$, let $\rdffrominstance(I)$ be defined as below.
\begin{align*}
  \rdftoinst(G, \typing) =&\{\Triple(s,p,o) \mid (s,p,o) \in G\}\\
  &\cup \{T(n) \mid n \text{ node of } G, T \in \typing(n)\}\\
  \rdffrominstance(I) =& (G, \typing) \text{ with }  G = \{(s,p,o)\mid \Triple(s,p,o) \in I \} \\
  & \text{ and } \typing(n) = \{T \in \Tsig \cup \{\LitSym\} \mid T(n) \in \Inst\} \text{ for any }n\text{ node of }G
\end{align*}
Then for any $\Tsig$-typed graph $(G, \typing)$ and any instance $\Inst$ of $(\GTsig, \emptyset)$ in the domain of $\rdffrominstance$, the following hold:
\begin{enumerate}
\item \label{item:one-to-one-rdftoinst} $\rdftoinst(G,\typing)$ is an instance of $(\GTsig, \emptyset)$;
\item \label{item:one-to-one-rdffrominst} $\rdffrominstance(\Inst)$ is a $\Tsig$-typed graph;
\item \label{item:one-to-one} $\rdffrominstance(\rdftoinst(G,\typing))$ is defined and is equal to $(G,\typing)$.
\end{enumerate}
\end{lemma}
\begin{proof} 
\begin{enumerate}
\item Immediately follows from the definition $\rdftoinst(G,\typing)$.
\item Immediately follows from the definition of $\rdffrominstance(\Inst)$.
\item 
Let $\Inst = \rdftoinst(G,\typing)$. 
By definition, $\rdffrominstance(\Inst)$ is defined if (a) $\Triple^\Inst \subseteq (\IriSet \cup \BlankSet) \times \IriSet \times (\IriSet \cup \BlankSet \cup \LitSet)$ and (b) $\LitSym^\Inst \subseteq \LitSet$ and (c) $T^\Inst \subseteq \IriSet$ for all $T \in \Tsig$.
Note that (a) follows from the definition of $\rdftoinst$ and the fact that $G$ is an RDF graph.
Also, (b) and (c) follow from the definition of $\rdftoinst$ and the fact that $\typing$ is a typing.
Then it immediately follows from the definitions that $\rdffrominstance(\rdftoinst(G,\typing)) = (G, \typing)$.

\end{enumerate}
\end{proof}

\subsection{Proof of Lemma~\ref{lem:shex-as-dependencies}}
 Take a typed graph $(G,\typing)$ and ShEx schema
    $\ShexSch = (\Tsig, \delta)$.  For the $\Rightarrow$ direction, we will prove by
    contrapositive. Assume that
    $(G,\typing)\not\models\shexdep$. Our goal is to
    prove $(G, \typing)$ is not correctly typed w.r.t. $\ShexSch$. By definition
    of entailment, there is one dependency $\sigma \in \shexdep$ that is not satisfied. The dependency $\sigma$ can be of the following forms:

    \begin{itemize}
    \item $\multleastone(T_s,p)$. By construction of $\shexdep$, the dependency $\sigma$
    occurs when a triple constraint is of the form $p\dbl T_o^\mu$ where $\mu\in \{\ONE,\PLUS\}$ and $p$ some property. Since $\sigma$ is not satisfied, $T_s \in \typing(n)$ for some node $n$ of $G$. Because the cardinalty of the set of triples with node $n$ and propery $p$ is 0, the definition
    of correctly typed in the typed graph $(G, \typing)$ w.r.t. $\ShexSch$ is violated.
    \item $\multmostone(T_s,p)$. By construction of $\shexdep$, the dependency $\sigma$
    occurs when a triple constraint is of the form $p\dbl T_o^\mu$ where $\mu\in \{\ONE,\MAYBE\}$. Since $\sigma$ is not satisfied, we have that $(s, p, o_1)\in G$
    and $(s, p, o_2)\in G$ and $T_s \in \typing(s)$, which violates the definition
    of correctly typed in the typed graph $(G, \typing)$ w.r.t. $\ShexSch$.
        \item $\tconstr(T_s,T_o,p)$. By construction of $\shexdep$, the dependency $\sigma$
    occurs when a triple constraint is of the form $p\dbl T_o^\mu$ where $\mu\in \{\ONE,\MAYBE,\MANY,\PLUS\}$. Since  $\sigma$ is not satisfied, $(s, p, o) \in G$ and $T_s \in \typing(s)$. Because the node $o \in G$, it must hold $T_o \in \typing(o)$. But this fact is not, then the typed graph $(G, \typing)$ w.r.t. $\ShexSch$ is not correctly typed.
    \end{itemize}   

    For the $\Leftarrow$ direction, assume that $(G,\typing)\models\shexdep$. Our goal is to prove $(G, \typing)$ is correctly typed w.r.t. $\ShexSch$. We will prove by contradiction. Suppose that $(G, \typing)$ is not correctly typed  w.r.t. $\ShexSch$. Then we have two cases when there is a node $n\in G$:
    \begin{itemize}
    	\item $\LitSym \in \typing(n)$ and $n\not \in \LitSet$. By definition of $\LitSym$, the node $n$ is of type literal, means $n\in \LitSet$. Contradiction.
    	\item We have two sub-cases when $T\in \typing(n)$: 
    	\begin{itemize}
    		\item $n\not \in \IriSet$. By definition, all nodes of $G$ are in the set $\LitSet\cup \IriSet \cup \BlankSet$. Because $T(n)$ is fact in $(G,\typing)$, then $n\in \IriSet\cup\BlankSet$. Because blank nodes are potentially IRIs, then $n\in \IriSet$. Contradiction.
    		\item There is a triple constraint $p\dbl{}S^\mu \in \delta(T)$  such that 
    		\begin{itemize}
    			\item There is a triple $(n,p,m)$ such that $S\not \in \typing(m)$. Since $T(n)$ and $\Triple(n,p,m)$ are facts in $(G,\typing)$ and $ (G,\typing)\models \tconstr(T,S,p)$, then $S(m)$ is fact in  $(G,\typing)$. Thus, $S \in \typing(m)$. Contradiction.
    			\item Let $K$ be the set of triples whose first element is $n$ and second element is $p$. The cardinality of $K$ is not bounded by $\mu$. Thus, we have the following cases:
    			\begin{itemize}
    				\item When $\mu=\ONE$ and $|K|\not=1$. It follows that $\multmostone(T,S,p) \in \shexdep$ and $\multleastone(T,S,p)\in \shexdep$. Since $(G,\typing)\models \shexdep$, then $|K|=1$. Contradiction.
    				\item When $\mu=\MAYBE$ and $|K|>1$. It follows that $\multmostone(T,S,p) \in \shexdep$. Since $(G,\typing)\models \shexdep$, then $|K|\leq 1$. Contradiction.    	
    				\item When $\mu=\PLUS$ and $|K|<1$. It follows that $\multleastone(T,S,p) \in \shexdep$. Since $(G,\typing)\models \shexdep$, then $|K|\ge 1$. Contradiction. 
    			\end{itemize}
    		\end{itemize}
    	\end{itemize}
    \end{itemize}

\section{The chase}
\label{app:universal}
Let $\DES = (\Rsch, \ShexSch, \stdep, \Fsig, \Fint)$ be a data exchange setting with $\Rsch = (\Rsig, \fdep)$ and $\ShexSch = (\Tsig, \delta)$, and let $\Inst$ be an instance of $\Rsig \cup \GTsig$.
For a tgd or std  $\sigma = \forall \vect{x}. \phi \to \psi$ and a homomorphism $h: \phi \to \Inst$, we say that $\sigma$ is applicable to $\Inst$ with $h$ if (1) either $\psi$ is without existential quantifier and $\Inst \cup \Fint, h \not\models \psi$, or (2) $\psi = \exists \vect{y}. \psi'$ and for all $h'$ extension of $h$ on $\vect{y}$, $\Inst \cup \Fint, h' \not\models \psi'$.
Then applying $\sigma$ to $\Inst$ with $h$ yields the instance $\Inst'$ defined as follows.
In the case (1), $\Inst' = h^\Fint(\psi)$.
In the case (2), $\Inst' = h'^\Fint(\psi')$ where $h'$ is an extension of $h$ and for $y \in \vect{y}$, $h'(y)$ is a fresh null value that depends on $\ShexSch$.
If $\delta(T)$ contains a triple constraint $p \dbl \LitSym^\mu$, then $h'(y) \in \NullLitSet \setminus \dom(\Inst)$.
If $\delta(T)$ contains $p \dbl T'^\mu$ for some $T' \in \Tsig$, then $h'(y) \in \BlankSet \setminus \dom(\Inst)$.
For an egd $\sigma = \forall \vect{x}. \phi \to x = x'$, if there exists a homomorphism $h: \phi \to \Inst$ s.t. $h(x) \neq h(x')$, we say that $\sigma$ is applicable to $\Inst$ with $h$ and the result is (1) the instance $\Inst'$ obtained by replacing $h(x)$ by $h(x')$ (resp. $h(x')$ by $h(x)$) in all facts of $\Inst$ if $h(x)$ (resp. $h(x')$) is a null value, and (2) the failure denoted $\bot$ if both $h(x)$ and $h(x')$ are non nulls.
We write $\chasestep{\Inst}{\sigma}{h}{U}$ if $\sigma$ is applicable to $\Inst$ with $h$ yielding $U$, where $U$ is either another instance or $\bot$, and  $\chasestep{\Inst}{\sigma}{h}{U}$ is called a \emph{chase step}.

Let $\Sigma$ be a set of dependencies and $\Inst$ be an instance.
A \emph{chase sequence} of $\Inst$ with $\Sigma$ is a finite or infinite sequence of chase steps $\chasestep{\Inst_i}{\sigma_i}{h_i}{\Inst_{i+1}}$ for $i = 0, 1,\ldots$, with $\Inst_0 = \Inst$ and $\sigma_i$ a dependency in $\Sigma$.
The well-known result from \cite{fagin:2005a} still holds in our setting: if there exists a finite chase sequence then it constructs a universal solution.
\section{Proofs of Theorems~\ref{thm:key-covered-fully-typed-iff-consistent} and \ref{thm:key-covered-decidable}}
\label{app:consistency}
Before proving the theorems, we define a mapping $h^F$ that will be used to define the notion of homomorphism from a formula into an instance. 
Let $\Fsig$ be a function signature and $F$ be an interpretation of $\Fsig$.
For a term $t$ over $\Fsig$ and a mapping $h: \Vars \to \Dom$, we define $h^F(t)$ as:
$$
h^F(t) =
\begin{cases}
  h(x) & \text{ if } t = x \in \Vars\\
  a & \text{ if } t = a \in \Dom\\
  f(h^F(\vect{t'})) & \text{ if } t = f(\vect{t'}) \text{ is a function term}.
\end{cases}
$$
The mapping $h^F$ is extended on atoms and conjunctions of atoms as expected: $h^F(R(\vect{t})) = R(h^F(\vect{t}))$ and $h^F(\bigwedge_{i \in 1..k} R_i(\vect{t}_i)) = \bigwedge_{i \in 1..k}h^F(R_i(\vect{t}_i))$.
Note that if the argument of $h^F$ does not contain function terms, the interpretation $F$ is irrelevant so we allow to omit the $F$ superscript and write e.g. $h(\vect{t})$ instead of $h^F(\vect{t})$.

A \emph{homomorphism} $h: \phi \to \Model$ between the conjunction of atoms $\phi$ over signature $\FOsig = \Rsig \cup \Fsig$ and the model $\Model = \Inst \cup F$ of $\FOsig$ is a mapping from $\fvars(\phi)$ to $\Dom$ s.t. for every atom $R(\vect{t})$ in $\phi$ it holds that $R(h^F(\vect{t}))$ is a fact in $\Inst$, where $\Inst$, resp. $F$, is the restriction of $\Model$ to $\Rsig$, resp. to $\Fsig$.

Remark that if $\phi$ does not contain function terms, then $F$ in the above definition is irrelevant and we write $h: \phi \to \Inst$ instead of $h: \phi \to \Model$ and $h(\vect{t})$ instead of $h^F(\vect{t})$.
\subsection{Proof of Theorem~\ref{thm:key-covered-fully-typed-iff-consistent}}


  Take a data exchange setting $\DES  = (\Rsch, \ShexSch, \stdep, \Fsig, \Fint)$ with $\ShexSch = (\Tsig, \delta)$.
  Assume first that $\DES$ is consistent, and let $\Inst$ be a valid instance of $\Rsch$ and $J$ be a solution for $\Inst$ by $\DES$.
  That is, $\Inst \cup J \models \stdep \cup \shexdep$.
  Let $T_s,T_o,p$ and $\mu \in \{\ONE, \MAYBE\}$ be such that $p\dbl{} T_o^\mu \in \delta(T_s)$.
  Suppose by contradiction that, for $i = 1,2$, $\sigma_i = \forall \vect{x}. \phi_i(\vect{x}_i, \vect{y}_i) \Rightarrow \Triple(f(\vect{x}_i), p, o_i) \wedge T_s(f(\vect{x}_i)) \wedge T_o(o_i)$ are two contentious stds in $\stdep$ and they are not functionally overlapped that is $\Inst \not\models \forall \vect{x}_1,\vect{x}_2, \vect{y}_1, \vect{y}_2. \phi_1(\vect{x}_1,\vect{y}_1) \wedge \phi_2(\vect{x}_2,\vect{y}_2) \wedge \vect{x}_1=\vect{x}_2 \Rightarrow o_1 = o_2$.
  That is, there is a homomorphism $h : \phi_1 \wedge \phi_2 \to \Inst$ s.t. $\Inst,h \models \phi_1 \wedge \phi_2$ but $h^\Fint(o_1) \neq h^\Fint(o_2)$.
  Because $J$ is a solution of $\DES$, we know that $\Inst \cup J \cup \Fint \models \sigma_i$ for $i=1,2$ and deduce that $J$ contains the facts (1) $\Triple(h^\Fint(f(\vect{x}_1)),p,h^\Fint(o_1))$, $\Triple(h^\Fint(f(\vect{x}_1)),p,h^\Fint (o_2))$ and $T_s(h^\Fint(f(\vect{x}_1)))$.
  On the other hand, by definition $\multmostone(T_s,p) = \forall x,y,z.\ T_s(x) \wedge \Triple(x, p, y) \wedge \Triple(x, p, z) \Rightarrow y = z$ is in $\shexdep$ and $J \models \multmostone(T_s,p)$.
  But $\multmostone(T_s,p)$ applies on the facts (1) with homomorphism $h'$ s.t. $h'(x) = h^\Fint(f(\vect{x}_1))$, $h'(y) = h^\Fint(o_1)$ and $h'(z) = h^\Fint(o_2)$, therefore $h^\Fint(o_1) = h^\Fint(o_2)$.
  Contradiction.

  Assume now that $\DES$ is key-covered, and let $\Inst$ be a valid instance of $\Rsch$.
  We construct a solution for $\Inst$ by $\DES$.
  We first chase $I$ with $\stdep$ until no more rules are applicable, yielding an instance $J$.
  Because $\stdep$ contains only stds (that is tgds on different source and target signatures), we know that $J$ exists.
  We now show that no egd from $\shexdep$ is applicable to $J$.
  By contradiction, let $\multmostone(T_s,p) = \forall x,y_1,y_2.\ T_s(x) \wedge \Triple(x, p, y_1) \wedge \Triple(x, p, y_2) \Rightarrow y_1 = y_2$ be an egd that is applicable to $J$.
  That is, there is a homomorphism $h: T_s(x) \wedge \Triple(x, p, y_1) \wedge \Triple(x, p, y_2) \to \Inst$ s.t. $\Triple(h(x), p, h(y_1))$, $\Triple(h(x), p, h(y_2))$ and $T_s(h(x))$ are facts in $J$ and $h(y_1) \neq h(y_2)$.
  By construction of $J$ as the result of chasing $\Inst$ with $\stdep$ and by the fact that $\stdep$ is fully-typed, it follows that there are two (not necessarily distinct) stds $\sigma_i = \forall \vect{x},\vect{y}_i. \phi_i(\vect{x}, \vect{y}_i) \Rightarrow \Triple(f(\vect{x}), p, o_i) \wedge T_s(f(\vect{x})) \wedge T_o(o_i)$ and there exist $h_i : \phi_i \to \Inst$ homomorphisms satisfying the following: (2) $f^\Fint(h_i(\vect{x})) = h(x)$, and $h_i(z_i) = h(y_i)$ if $o_i = z_i$ are variables, and $g^\Fint(h_i(\vect{z}_i)) = h(y_i)$ if $o_i = g(\vect{z}_i)$ for some vectors of variables $\vect{z}_i$ and function symbol $g$, for $i=1,2$.
  Then $h_1 \cup h_2: \phi_1 \wedge \phi_2 \to \Inst$ is a homomorphism, and because $\DES$ is key-covered we know that $h_1(o_1) = h_2(o_2)$.
  This is a contradiction with $h(y_1) \neq h(y_2)$ using (2) and the fact that the functions $f^\Fint$ and $g^\Fint$ are injective, and implies that no egd from $\shexdep$ is applicable to $J$.

  Finally, we are going to add the facts $J'$ to $J$ so that $J \cup J'$ satisfies the tgds and the egd's in $\shexdep$.
  Note that $J$ does not satisfy $\shexdep$ because some of the $\multleastone(T_s,p)$ might not be satisfied.
  For any $\multleastone(T_s,p)$ in $\shexdep$, let $b^{T_s,p}\in \BlankSet$ be a blank node distinct from other such blank nodes, that is, $b^{T_s,p} \neq b^{T_o,p'}$ if $T_s \neq T_o$ or $p \neq p'$.
  Now, let $J_1$ and $J_2$ be the sets of facts defined by:
  \begin{align*}
  J_1 &= \left\{    
    T_1(b^{T_s,p}) \mid p\dbl T_1^\mu \in \delta(T_s) \text{ for } \mu \in \{\ONE,\PLUS\}
    \right\}\\
  J_2 &= \left\{    
    \Triple(b^{T_s,p}, q, b^{T_1,q}) \mid T_1(b^{T_s,p}) \in J_1 \text{ and }  q\dbl T_2^\mu \in \delta(T_1) \text{ for } \mu \in \{\ONE,\PLUS\}
    \right\}
  \end{align*}
Intuitively, $J_1$ adds to the graph nodes $b^{T_s,p}$ whenever the property $p$ is required by type $T_s$ in $\ShexSch$.
A property is required if it appears in a triple constraint with multiplicity $\ONE$ or $\PLUS$.
Such node has type $T_1$ as required by the corresponding triple constraint $p\dbl{}T_1^\mu$ in $\delta(T_s)$.
Then, $J_2$ adds to the graph triples for the properties $q$ that are required by the nodes added by $J_1$.
Remark that $J_1 \cup J_2$ is a correctly typed graph.
We finally connect $J_1 \cup J_2$ to $J$.
Let 
$$
J_3 = \left\{
  \Triple(a, p, b^{T_s,p}) \mid T_s(a) \in J \text{ and } \not\exists \Triple(a,p,b') \text{ in } J \text{ and } p\dbl{}T_o^\mu \in \delta(T_s) 
\right\}
$$
Then $G = J \cup J_1 \cup J_2 \cup J_3$ satisfies the tgds in $\shexdep$.
It remains to show that $G$ also satisfies the egd's in $\shexdep$.
This is ensured by construction as $J$ satisfies the egd's and $J_2$ and $J_3$ add a unique triple $\Triple(b, p, b')$ only to unsatisfied typing requirements $T_s(b)$ for types $T_s$, that is, for every  $\Triple(b,p,b')$ added by $J_2$ or $J_3$ there is no different $\Triple(b,p,b'')$ in $J \cup J_2 \cup J_3$.

This concludes the proof of Theorem~\ref{thm:key-covered-fully-typed-iff-consistent}.

\subsection{Proof of Theorem~\ref{thm:key-covered-decidable}}

Let $\DES = (\Rsch, \ShexSch, \stdep, \Fsig, \Fint)$ with $\ShexSch = (\Tsig, \delta)$ and $\Rsch = (\Rsig, \fdep)$ be a fully-typed data exchange setting.

The proof goes by reduction to the problem of functional dependency propagation.
We start by fixing some vocabulary and notions standard in databases.
A \emph{view} over a relational signature $\Rsig$ is a set of queries over $\Rsig$.
Recall that a $n$-ary query is a logical formula with $n$ free variables.
If $\Vsig = \{V_1, \ldots, V_n\}$ is a view, we see $\Vsig$ as a relational signature, where the arity of the symbol $V_i$ is the same as the arity of the query $V_i$, for $1 \le i \le n$.
Given a relational schema $\Rsch = (\Rsig, \fdep)$, a view $\Vsig$, and an instance $\Inst$ of $\Rsch$, by $\Vsig(\Inst)$ we denote the result of applying the query $\Vsig$ to $\Inst$.
The latter is an instance over the signature $\Vsig$.

Now, the problem of functional dependency propagation $\fdprop(\Rsch, \Vsig, \fdep^{\Vsig})$ is defined as follows.
Given a relational schema $\Rsch = (\Rsig, \fdep)$, a view $\Vsig$ over $\Rsig$, and a set of functional dependencies $\fdep^\Vsig$ over $\Vsig$, $\Rsch = (\Rsig, \fdep)$ holds iff for any $\Inst$ valid instance of $\Rsch$, $\Vsig(\Inst) \models \fdep^\Vsig$.
It is known by \cite{klug:1982a} that the problem $\fdprop(\Rsch, \Vsig, \fdep^\Vsig)$ is decidable.

We will construct a view $\Vsig$ and a set $\fdep^\Vsig$ of functional dependencies over $\Vsig$ s.t.  $\fdprop(\Rsch, \Vsig, \fdep^\Vsig)$ iff $\DES$ is key-covered.

Let $\sigma_1, \sigma_2$ be two contentious stds from $\stdep$ that are functionally overlapping as those in the premise of the key-coverdness condition.
That is, for some $T_s,T_o,p,f$, for $i=1,2$, we have $\sigma_i = \forall \vect{x}_i, \vect{y}_i. \phi_i(\vect{x}_i, \vect{y}_i) \Rightarrow \Triple(f(\vect{x}_i), p, o_i) \wedge T_s(f(\vect{x}_i)) \wedge T_o(o_i)$.
Recall that $o_i$ and $o_i$ and either both variables, or are both functional terms with the same function symbol.
Let $\vect{z}$, resp. $\vect{z}'$ be the vectors of variables is $o_1$, resp. $o_2$.
That is, if e.g. $o_1$ is a variable then $\vect{z}$ is a vector of length one of this variable, and if $o_1 = g(z_1, \ldots, z_n)$ for some function symbol $g$, then $\vect{z} = z_1, \ldots, z_n$.
Remark that $\vect{z} \subseteq \vect{x} \cup \vect{y}_1$, and similarly for $\vect{z'}$.

Now, for any such couple $\sigma_1, \sigma_2$ of two (not necessarily distinct) stds, we define the query $V_{\sigma_1,\sigma_2}$ as the union of two queries, and the functional dependency $\fd_{\sigma_1,\sigma_2}$, as follows.
\begin{align}
  V_{\sigma_1}(\vect{x}, \vect{z}) &= \exists \vect{y}_1^{-\vect{z}}\ . \phi_1(\vect{x}, \vect{y}_1)\\
  V_{\sigma_2}(\vect{x}, \vect{z'}) &= \exists \vect{y}_2^{-\vect{z}'}\ . \phi_2(\vect{x}, \vect{y}_2)\\
  V_{\sigma_1, \sigma_2} &= q_{\sigma_1} \cup q_{\sigma_2}\\
  \fd_{\sigma_1, \sigma_2} &= V_{\sigma_1, \sigma_2}:\{1,\ldots,m\} \to \{m+1, \ldots, m+n\}
\end{align}
where for any two vectors of variables $\vect{y}$ and $\vect{z}$, $\vect{y}^{-\vect{z}}$ designates the set of variables $\vect{y} \setminus \vect{z}$, and $m$ is the length of $\vect{x}$, and $n$ is the length of $\vect{z}$ and $\vect{z}'$.
Then
\begin{align}
  \Vsig &= \left\{V_{\sigma_1, \sigma_2} \mid \sigma_1, \sigma_2 \text{ as in the premise of the condition for key-covered}\right\}\\
  \fdep^\Vsig &= \left\{\fd_{\sigma_1, \sigma_2} \mid \sigma_1, \sigma_2 \text{ as in the premise of the condition for key-covered}\right\}
\end{align}

The sequel is the proof that $\fdprop(\Rsch, \Vsig, \fdep^\Vsig)$ iff $\DES$ is key-covered, which by \cite{klug:1982a} implies that key-coverdness is decidable.

For the $\Rightarrow$ direction, suppose that $\fdprop(\Rsch, \Vsig, \fdep^\Vsig)$.
Let $\Inst$ a valid instance of $\Rsch$ and let $J = \Vsig(I)$.
We show that for any two contentious stds $\sigma_1,\sigma_2 \in \stdep$ that are functionally overlapping as in the premise of the condition for key-covered, it holds that $\Inst \cup \Fint \models \forall \vect{x}_1,\vect{x}_2, \vect{y}_1, \vect{y}_2. \phi_1(\vect{x}_1, \vect{y}_1) \wedge \phi_2(\vect{x}_2, \vect{y}_2) \wedge \vect{x}_1=\vect{x}_2 \Rightarrow o_1 = o_2$.
Let $\nu$ be a valuation of the variables $\vect{x}_1 \cup \vect{y}_1 \cup \vect{y}_2$ s.t. $\Inst \cup \Fint, \nu \models \phi_1 \wedge \phi_2$.
By definition of $q_{\sigma_1}$ and $q_{\sigma_2}$ and $\vect{x}_1=\vect{x}_2$ it is easy to see that $V_{\sigma_1, \sigma_2}(\nu(\vect{x}_1), \nu(\vect{z}))$, and $V_{\sigma_1, \sigma_2}(\nu(\vect{x}_2), \nu(\vect{z}'))$ are facts in $J$
Because $J$ satisfies $\fd_{\sigma_1,\sigma_2}$, we deduce that $\nu(\vect{z}) = \nu(\vect{z}')$, therefore $\nu^{\Fint}(o_1) = \nu^{\Fint}(o_2)$, which concludes the proof of the $\Rightarrow$ direction.

\smallskip
For the $\Leftarrow$ direction, suppose that $\DES$ is key-covered.
Let $\Inst$ a valid instance of $\Rsch$ and let $J = \Vsig(I)$.
Let $\sigma_i = \forall \vect{x}_i, \vect{y}_i. \phi_i(\vect{x}_i, \vect{y}_i) \Rightarrow \Triple(f(\vect{x}_i), p, o_i) \wedge T_s(f(\vect{x}_i)) \wedge T_o(o_i)$ for $i = 1,2$ be two stds in $\stdep$ that satisfy the premise for key-covered.
Let $V_{\sigma_1, \sigma_2}(\vect{a},\vect{b})$ and $V_{\sigma_1, \sigma_2}(\vect{a},\vect{b}')$ be two facts in $J$.
That is, by definition and $\vect{x}_1=\vect{x}_2$ there exist valuations $\nu$ of the variables $\vect{y}_1^{-\vect{z}}$ and $\nu'$ of the variables $\vect{y}_2^{-\vect{z'}}$ s.t. $I, \nu[\vect{x}_1/\vect{a},\vect{z}/\vect{b}] \models \phi_1(\vect{x}_1, \vect{y}_1)$ and  $I, \nu'[\vect{x}_2/\vect{a},\vect{z}'/\vect{b'}] \models \phi_2(\vect{x}_2, \vect{y}_2)$.
We now distinguish two cases, depending on whether the two facts were generated by the same query $V_{\sigma_i}$ (for some $i \in 1..2$), or one was generated by $V_{\sigma_1}$ and the other one by $V_{\sigma_2}$.
\begin{itemize}
\item If $(\vect{a},\vect{b}) \in V^J_{\sigma_1}$ and $(\vect{a},\vect{b}') \in V^J_{\sigma_2}$, then $I \cup \Fint, \nu\cup\nu'\cup[\vect{x}_1/\vect{a},\vect{z}/\vect{b},\vect{z'}/\vect{b'}] \models \phi_1(\vect{x}_1, \vect{y}_1) \wedge \phi_2(\vect{x}_2, \vect{y}_2)$, where $\vect{x}_1=\vect{x}_2$.
  Thus, because $\DES$ is key-covered we know that $[\vect{z}/\vect{b}]^{\Fint}(o_1) = [\vect{z'}/\vect{b'}]^{\Fint}(o_2)$, so $\vect{b} = \vect{b'}$.
  Therefore $J \models \fd_{\sigma_2, \sigma_2}$.
\item If $(\vect{a},\vect{b}), (\vect{a},\vect{b}') \in V^J_{\sigma_1}$, then by definition of the view $\Vsig$ it is easy to see that $V_{\sigma_1,\sigma_1}(\vect{a},\vect{b})$ and $V_{\sigma_1,\sigma_1}(\vect{a},\vect{b}')$ are also facts in $J$.
  Then the proof goes as in the previous case.
\end{itemize}
This concludes the proof of Theorem~\ref{thm:key-covered-decidable}.

\end{document}